\documentclass[12pt,leqno]{article}

\newcommand{\subf}[2]{%
  {\small\begin{tabular}[t]{@{}c@{}}
  #1\\#2
  \end{tabular}}%
}

\usepackage{array}
\usepackage[dvipsnames]{xcolor}
\usepackage{tikz}

\usepackage{amsmath}

\numberwithin{equation}{section}
\usepackage{graphicx}
\usepackage{setspace}

\usepackage{authblk}

\usepackage{thmtools, thm-restate}

\newtheorem{theorem}{Definition}

\newenvironment{proof}[1][Proof]{\begin{trivlist}
\item[\hskip \labelsep {\bfseries #1}]}{\end{trivlist}}

\newcolumntype{M}[1]{>{\centering\arraybackslash}m{#1}}

\begin{document}


\author[1, 2, 3]{\small David B. McMillon}
\affil[1]{\footnotesize Department of Economics, Emory University}
\affil[2]{\footnotesize Federal Reserve Bank of Atlanta}
\affil[3]{\footnotesize Stone Center for Research on Wealth Inequality and Mobility}



\title{What Makes Systemic Discrimination, ``Systemic?" Exposing the Amplifiers of Inequity [Draft]} 
\maketitle


\thispagestyle{empty}






\abstract{

\small



Drawing on work spanning economics, public health, education, sociology, and law, I formalize theoretically what makes systemic discrimination “systemic."  Injustices do not occur in isolation, but within a complex system of interdependent factors; and their effects may amplify as a consequence. I develop a taxonomy of these amplification mechanisms, connecting them to well-understood concepts in economics that are precise, testable and policy-oriented.  This framework reveals that these amplification mechanisms can either be directly disrupted, or exploited to amplify the effects of equity-focused interventions instead.  In other words, it shows how to use the machinery of systemic discrimination against itself.  Real-world examples discussed include but are not limited to reparations for slavery and Jim Crow, vouchers or place-based neighborhood interventions, police shootings, affirmative action, and Covid-19.} 

\vspace{0.5 cm}


\vspace{0.5 cm}




\small


\clearpage 
 \pagenumbering{arabic}

\doublespacing

\textit{No problem can be solved from the same level of consciousness that created it.}
\vspace{0.3cm}

-\textbf{Albert Einstein}

\vspace{0.3cm}

\section{Introduction}




Amidst rising public discourse on the persistence, prevalence, and causes of American racial inequities rings a salient and polarizing phrase: ``systemic discrimination."  Political, legal, and academic conversations surrounding ``systemic discrimination" are burdened by misunderstanding and controversy, as with many definitionally amorphous terms.  This lack of precision has impeded the kind of scientific progress that can lead to coordinated, calculated, and sustained reductions in otherwise large and remarkably persistent inequities in well-being (Bailey, Z.D., et al., 2017; Darity Jr, W.A. et al., 2022; Derenoncourt, E. et al., 2022; Hamilton, D. et al., 2010; Hoover et al., 2021; Miller, M.C., 2020; Roithmayr, D., 2014).

Social scientists' contributions to the conceptualization of systemic discrimination span multiple fields.  In law, for example, Powell (2007) discusses the interactions of institutions that, intentionally or not, reinforce racialized outcomes.  Historians like Feagin (2013) have viewed systemic discrimination as a system of intentional exploitation: “The complex array of antiblack practices, the unjustly-gained political-economic power of whites, the continuing economic and other resource inequalities along racial lines, and the white racist attitudes created to maintain and rationalize white privilege and power.”   In contrast, sociologists such as Reskin (2012) have emphasized a self-reinforcing system of race-linked disparities that can propagate across sectors regardless of intent.  This has been argued in the public health literature as well, which asserts that the way people make meaning out of these disparities can generate discrimination endogenously.  That is, ``these patterns and practices in turn reinforce discriminatory beliefs, values, and distribution of resources" (Bailey et al., 2017).

The field of economics has lagged behind public health, sociology, and law in proposing theories of systemic discrimination.  Traditionally, economists have focused on taste-based discrimination (in which individuals or firms derive utility from engaging in a discriminatory injustice), statistical discrimination (accurate or inaccurate, group-based assumptions about individuals of a minority group due to a lack of information), and alternative explanations for inequality besides discriminatory injustices, (Charles, K.K., 2011) including self-fulfilling prophecies, social signification, and racial stigma (Loury, G.C., 2009).  The field has also seen extensive work towards economic equity in labor, education, criminal justice, and other particular sectors (Jones, D. et al., 2022; Miller, C., 2017; Johnson, R.C. et al., 2019; Ba, B., 2021; Cunningham, J. P. et al., 2021; Cook, L.D., 2020).  This work is consistent with the approach of stratification economics, which focuses on how racial disparities in material conditions are caused by present-day and historical intergroup injustices (Darity Jr, W. A. 2022).  More recently, economists have finally begun considering systemic discrimination: ``[Indirect] discrimination emerging from group-based differences in non-group characteristics" (Bohren et al., 2022).  This definition is precise and econometrically tractable, but much narrower than the notions considered in other social sciences.  Moreover, while this definition has been taken up (e.g., Zivin et al., 2023), ``systemic discrimination" is still being used more loosely in the field of economics (e.g., as a practice that leaves a broad impact; Kline et al., 2022).  


Despite disagreement on the definition of systemic discrimination across the social sciences, the phrase is generally intended to help explain the prevalence and persistence of inequities.  Perhaps the key to harnessing scientific thought to fight those inequities lies in understanding what exactly makes systemic discrimination ``systemic."  One theme that is consistent across the above definitions is that initial effects of an injustice are allowed to amplify in some way.

Specifically, as discussed in the subfield of complexity economics (Arthur, W.B.; Durlauf, S. N., 2012), the social systems of interest can exhibit features of what are called complex systems (Reskin, 2012; McMillon et al., 2014; Roithmayr, 2014).  The effects of an injustice can amplify contemporaneously (spillover across outcomes or sectors, social multipliers, or synergies between inequities and shocks), or temporally (persistence through reinforcement, e.g., across generations), regardless of intent, because such scaling is a mathematical consequence of the interdependent nature of complex systems.  This view includes but is not limited to the definition proposed by Bohren and coauthors (Bohren et al., 2022), and addresses major concerns emphasized in the economics of discrimination (Lang, K. et al., 2020; Charles, K.K., 2011). This paper is the first to formalize systemic discrimination as a phenomenon in which the initial inequities from an injustice are amplified by a system's interdependencies.  In no way does this perspective absolve discriminatory behavior of responsibility--especially since the initial inequities still stem from injustices, and since amplification mechanisms can be introduced intentionally.  It does something much more important: it points us to the kinds of technical tools we will need to implement policy solutions that yield large persistent effects on racial equity.  


This view of systemic discrimination produces a key insight: that the very features that amplify inequities in the status quo can be harnessed to amplify the effects of equity-focused interventions as well.  In other words, it shows how to use the machinery of systemic discrimination against itself.  When faced with an inequity and a mechanism that amplifies it, we can either ``disrupt" the amplification mechanism directly, or ``exploit" it, amplifying the effects of interventions that rectify the initial inequity like a booster engine.  This paper develops the first formalized taxonomy of such amplification mechanisms and connects them to well-understood concepts in economics.  The result is a practical, precise, and unified framework through which social scientists can not only identify and measure, but more strategically and efficiently combat systemic discrimination.  This provides a clear comparative advantage over the generic view of systemic discrimination as a widespread practice leaving a broad impact.

The remainder of this paper proceeds as follows.  Section $2$ presents a stochastic model of systemic discrimination.  Section $3$ presents a taxonomy of amplification mechanisms embedded within the model in Section $2$, along with policy examples.  Most proofs are in the Appendix.  Section $4$ discusses limitations, implications for measurement, testable predictions, and policy.  The final section concludes.

\section{Model}

 The contributions of this paper regard systemic discrimination in general, including, for example, systemic gender discrimination.  However, I will focus primarily on systemic racial discrimination for substantive context.   The contributions of this paper are also independent of what is considered discriminatory or unjust, which will be intentionally left up to the reader.  Importantly, a racial \textbf{inequity} will be defined as a racial inequality stemming from an injustice.  This framing forces the reader to focus on racial differences that are unjust by her own admission, regardless of the extent to which she believes some racial differences are permissible, justifiable, statistical normalities or even natural.  It ensures a productive conversation about systemic discrimination, a politically polarizing topic, by focusing on the amplification of what the reader, regardless of her political perspective on race, necessarily acknowledges as immoral.  
 
 This definition of `inequity' is closely related to definitions proposed in the public health literature, such as by Whitehead and coauthors (2006).  However, the ``injustice" in this paper's definition is not necessarily due to racial inequalities that stem from group differences in the distribution of resources.  Instead, it derives from specific discriminatory acts, whether at the individual level or the governmental level.  These acts may well induce group differences in the distribution of resources, but the injustice lies in how those group differences arose, not in their mere existence.  This paper will remain agnostic about whether general inequality is inherently unjust.  Its focus is on how the initial effects of unjust treatment---inequities---amplify in various ways.  Unlike inequalities, inequities are, in this paper, unjust by definition.  The focus of this model regarding systemic discrimination is not inequality in general, but inequity in particular.  Methods and theoretical ramifications for distinguishing between inequalities that we find permissible or not, and when it is infeasible to do so, is discussed elsewhere (Jackson, J.W., 2021).  We will proceed under the assumption that we have identified inequities in particular.  

For convenience, I will measure ``inequities" based on the unjust, standardized distance of a disadvantaged person's outcome from the advantaged group's mean.  The purpose for not simply comparing group means is that the model needs to keep track of individuals, not only to illustrate how systemic discrimination can occur at the individual level, but also to illustrate how interactions between individuals lead to social multipliers.  ``Disadvantaged" in this context refers to the group against whom the injustice was carried out.  ``Advantaged" refers to a group against whom the injustice was not carried out.  Notice it is not about which group carried out an injustice.  This allows, for example, the consideration of educational inequities faced by descendants of slaves relative to Asian Americans.  

Finally, this model focuses only on temporary injustices, or injustices that can continually arise endogenously from the effects of previous temporary injustices.  It does not account for permanent, exogeneous injustices.  For example, if there is some ongoing, permanent, ``natural" level of racial animus that continues to impact Black Americans, it will not be addressed in this paper.  However, the model could easily be extended to consider this with the inclusion of a baseline nonzero constant level of inequity.  I omit this nuance in the interest of simplicity and clarity of focus, as the notion of systemic discrimination as widespread, embedded and permanent has already been studied.  Moreover, the importance of the arguments of this paper is only enhanced by the presence of such forces. 

To formalize the amplification mechanisms that make systemic discrimination, ``systemic," we must understand what the world would look like in the absence of amplification.  As a working example, suppose no further injustices occurred after emancipation in 1865.  Consider a world in which there were no mechanisms (including compound interest) that amplified racial wealth inequities (contemporaneously or temporally) thereafter.  How would racial wealth inequities have evolved?  We still wouldn't expect them to disappear completely the day the enslaved were freed, but we would expect them to decay relatively quickly, in some ``natural way."  Similarly, there would be other inequities (e.g., in education and health) in 1865 that would decay naturally over time.  Let us proceed by formalizing these notions.



What happens $t$ periods following an injustice?  Assume $t \geq 0$ and let $X(t) \in R^{NxM}$ be a random $N$ X $M$ matrix of normalized inequities for a minority group, such that entry $x_{ij}(t)$ represents the normalized value of inequity $j$ for person $i$ at time $t$, where $j=1,...,M$ and $i=1,...,N$.  Again, the inequity values indicate the standardized distance of person $i$ from the advantaged group's mean.  In that case, $x_{ij}=2$ means that person $i$ is two standard deviations below the White mean value of inequity $j$.  Finally, for the sake of consistency, assume all inequities are measured ``optimistically."  For example, inequities in heart disease should be measured as inequities in a ``lack" of heart disease, and inequities in crime should be measured as inequities in law-abiding behavior. \vspace{0.3cm}

 Assume $X(0)$ is given, and that for all $i=1,...N$  and $j=1,...M$, 
 $$x_{ij}(t)=\delta x_{ij}(t-1)+\epsilon_{ij}(t)$$
 $$\Rightarrow x_{ij}(t)=\delta^{t} x_{ij}(0)+ \sum_{k=1}^t \delta^{t-k} \epsilon_{ij}(k)$$

 for some idiosyncratic shock $\epsilon_{ij}(t)$ with a mean of zero and finite variance, and one-step decay factor $\delta \in (0,1)$.  Hence $(X(t))_{t \geq 0}$ constitutes a random process.  The decay factor should be conceptualized in a similar fashion as in macroeconomic models in which physical capital undergoes a natural decay over time.  It can be shown that if $E[\epsilon_{ij}(t)]=0$, the random process $(X(t))_{t \geq 0}$ should naturally lead to ``asymptotic equity" over time, in the sense that 

$$\lim_{t \rightarrow \infty} E[x_{ij}(t)]=0$$

for all $i=1,...N$ and $j=1,...M.$  That is because it is a Markov process with a decay factor $\delta \in (0,1)$.  Call $(X(t))_{t \geq 0}$ the \textbf{``asymptotically equitable"}  random process with decay factor $\delta$.

This asymptotically equitable random process describes how the world evolves following an injustice in the absence of amplification-and, importantly, in the absence of any other exogenously imposed injustices.  This model asserts that if there is no convergence to equity, it is necessarily due to some mechanism which can be captured as ``systemic" that belongs in the construction of $S$.  Second, convergence to equity in the absence of ``systemic discrimination" does not rule out racial inequality in general.  Again, from the specified definitions, we can reach equity (so that the effects of injustices have faded out) without equality, to the extent the remaining intergroup inequality is completely unrelated to injustices.

I will now introduce the ``system" $S$ that acts on this process.  Note that both of the following definitions indicate that $S$ is implicitly causal: $S$ is discriminatory if it causes initial inequities to amplify in some way.  Since there are many possible ways to produce a summary measure of inequities across individuals, I will use a generalizing, abstract matrix norm $\mu$.

\begin{theorem}
\textbf{Discriminatory in the long run}: A function

$$S: (X(t))_{t \geq 0} \rightarrow (Y(t))_{t \geq 0}$$

acting on an ``asymptotically equitable" random process $(X(t))_{t \geq 0}$ to generate the random process $(Y(t))_{t \geq 0}$ is \textbf{discriminatory in the long run} with respect to a matrix norm $\mu$ iff 

$$\lim_{t \rightarrow \infty} E[\mu(X(t))]<\lim_{t \rightarrow \infty} E[\mu(Y(t))]$$

\end{theorem}

That is, if $S$ generates inequity in the long run.  For example, if agglomeration externalities are sufficiently strong in a spatial economy, discriminatory shocks could have permanent, ``path dependent" consequences (Allen, T., et al., 2020).  However, inequities can be amplified temporarily as well, regardless of what happens in the steady state.  This brings us to the next definition:

\begin{theorem}
 \textbf{Discriminatory on an Interval}: A function

$$S: (X(t))_{t \geq 0} \rightarrow (Y(t))_{t \geq 0}$$

acting on an``asymptotically equitable" random process $(X(t))_{t \geq 0}$ is \textbf{discriminatory on the interval} $B$ with respect to a matrix norm $\mu$ iff 

$$ E[\mu(X(t))] < E[\mu(Y(t))]$$

for all $t \in B$. 

\end{theorem}

That is, if $S$ generates inequity during time interval $B$.\\

It is possible that a system can be discriminatory on an interval and not be discriminatory in the long run.  This is the case with bottlenecks, whereas although racial equity may be achieved in the steady-state, the time to converge to the equitable steady state could be extremely long.  For example, Derenoncourt and coauthors (Derenoncourt et al., 2022) show that even under equal, optimistic wealth-generating conditions, the racial wealth gap will take more than $200$ years to close.  This is a justification for reparations, and a similar justification has been proposed for affirmative action (Bagenstos, S. R., 2014), despite the recent supreme court ruling.  Bottlenecks matter regardless of the long-term result.

The function $S$ embeds a non-exhaustive ``taxonomy" of amplification mechanisms, described in the next section.  These mechanisms, I argue, are examples of what makes systemic discrimination ``systemic."  They exist precisely because injustices operate on interconnected systems.











\section{What is $S$? A Taxonomy}

\textit{“All things appear and disappear because of the concurrence of causes and conditions. Nothing ever exists entirely alone; everything is in relation to everything else.”}
\vspace{0.3cm}

-\textbf{The Buddha}\\

 $S$ describes how the interconnectedness of social, economic, or even biological systems amplify the initial effects of injustices.  It is an abstract function that acts on an asymptotically equitable random process $(X(t))_{t \geq 0}$, and transforms it into a random process $(Y(t))_{t \geq 0}$ that is either asymptotically inequitable, temporarily inequitable, or both.   $S$ can embed a myriad of well-known economic models that contain features that amplify inequities.  For instance, $S$ can represent poverty traps due to dynamic reinforcement between inequities (Durlauf, S. et al., 2017).  $S$ can also embed a system describing the dynamics of intergenerational mobility (Chetty, R. et al., 2020), the dynamics of a spatial economy  (Allen, T. et al., 2020), or the reinforcing dynamics of academic skill formation (Cunha, F. et al., 2007; McMillon, 2024).  But $S$ is also general enough to represent well-known phenomena across all sciences, including but not limited to the biological mechanisms that make those with heart disease susceptible to Covid-19 (Wadhera, R. K. et al., 2021), the psychological mechanisms behind stereotype threat (Spencer, S. J. et al., 2016), the sociological mechanisms through which police officers' use of force depends on peers (Roithmayr, D. et al., 2016), and even the political economy endogenizing the response of a dominant group to equity-focused interventions (Derenoncourt, E., 2022).

I will now describe four major ways that $S$ can amplify inequities.  Amplification mechanisms are the ``engines" of systemic discrimination.  In the interest of equity, we can either disrupt the engines themselves, or exploit them as engines for equity, as illustrated in the next section.

\subsection{Intersectoral Spillover}

 $S$ can induce $\textbf{Intersectoral Spillover}$, so that one inequity spills over into another.  Suppose that, due to past injustices, Black Americans are more likely to have a criminal history.  Then, even without direct discrimination in the labor market, those inequities will spill over into the labor market (Agan, A., and Starr, 2018), because employers are less likely to hire people with a criminal history.   In this case the amplification is about the propagation of inequities across different outcomes or sectors of the economy.  This propagation temporarily amplifies the initial effects of an injustice.  More precisely:

\begin{theorem}

Consider the abstract function $S: (X(t))_{t \geq 0} \rightarrow (Y(t))_{t \geq 0}$ acting on an asymptotically equitable random process $(X(t))_{t \geq 0}$ to generate random process $(Y(t))_{t \geq 0}$.  Let $\delta \in (0,1)$ be the discount factor that defines $(X(t))_{t \geq 0}$.\\

 \textbf{Intersectoral Spillover} occurs at time $t$ if for some $s<t$ and inequity $j$, there exists $b_k > 0$, $k \neq j$, s.t.
$$y_{ij}(t)=\delta^{t-s} y_{ij}(s)+b_k y_{ik}(s)+\epsilon_{ij}(t)$$

That is, inequity $j$ is modified by some other inequity $k \neq j$.\\

\end{theorem}

In this definition, one inequity $y_{ij}(t)$ is being influenced by the value of another inequity $y_{ik}(s)$ from an earlier time step $s<t$.  It is also influenced by its earlier value, discounted.  This discounting would also have happened in the asymptotically equitable process $(X(t))_{t \geq 0}$, but in the amplified process $(Y(t))_{t \geq 0}$, there is also spillover across inequities.  As a result, there is, at least temporarily, more expected inequity at time t in the amplified process than there would have been in $(X(t))_{t \geq 0}$.




\begin{restatable}{proposition}{Spillover}

Intersectoral Spillover is discriminatory over an interval.

\end{restatable}

(Proof in Appendix)\\

What does this mean for policy?  Let $y_{ik}(s)$ be the inequity in clean criminal history at time $s$, and $y_{ij}(t)$ be the inequity in hiring at time $t$.  With public policy, we can disrupt the amplification mechanism by reducing $b_k$.  For example, we can get employers to care less about whether applicants have a criminal history.  Note this is not the same as restricting information employers have, as in the case of ``Ban-the-Box."  

We can also exploit the amplification mechanism by reducing $y_{ik}(s)$.  That is, we can reduce the incidence of Black people being unjustly arrested or convicted.  In that case, a large $b_k$ would actually help ensure this intervention would have effects that spill over into the labor market, reducing $y_{ij}(t)$.

In the long run, however, both inequities still face decay factor $\delta$, and the expected inequities will fade over time even if one spills over into the other.  Intersectoral Spillover alone is not enough to cause the long-run expected inequities under the amplified process $(Y(t))_{t \geq 0}$ to exceed those under $(X(t))_{t \geq 0}$.  

\begin{restatable}{proposition}{Spilloverlong}

Intersectoral Spillover alone is not discriminatory in the long run.

\end{restatable}

(Proof in Appendix)\\

\subsection{Intersectoral Synergy}

$S$ can induce $\textbf{Intersectoral Synergy}$, so that the effect of a shock operating on one inequity is enhanced by the pre-existing presence of another inequity.  In other words, pre-existing inequities can make a group vulnerable to shocks.  For instance, consider racial wealth inequities in the U.S.  Because low wealth makes it more difficult to smooth consumption (Ganong, P. and Jones, D., 2020), Black Americans tend to be more heavily impacted by housing crises and other macroeconomic shocks (Hoover et al., 2021).  Another example is that racial inequities in heart disease increase systemic risk for covid-19 among Black Americans (“syndemics;" see Mendenhall, E., and Singer, 2018; Moore et al., J.T., 2020; and Wadhera, R. K. et al., 2021).   In each of these cases, amplification is happening because the inequities in one sector or outcome are amplifying systemic risk in another.

\begin{theorem}

Consider the abstract function $S: (X(t))_{t \geq 0} \rightarrow (Y(t))_{t \geq 0}$ acting on an asymptotically equitable random process $(X(t))_{t \geq 0}$ to generate random process $(Y(t))_{t \geq 0}$.  Let $\delta \in (0,1)$ be the discount factor that defines $(X(t))_{t \geq 0}$.\\

 \textbf{Intersectoral Synergy} occurs at time $t$ if for some $s<t$ and inequity $j$, the following are true:
 
 \begin{enumerate}
 
 \item If $\epsilon_{ij}(t)>0$, there exists $c_k>0$, $k \neq j$, s.t.
$$y_{ij}(t)=\delta^{t-s} y_{ij}(s)+c_k y_{ik}(s) \epsilon_{ij}(t)+\epsilon_{ij}(t)$$

\item If $\epsilon_{ij}(t)<0$, there exists $c_k<0$, $k \neq j$, s.t.
$$y_{ij}(t)=\delta^{t-s} y_{ij}(s)+c_k y_{ik}(s) \epsilon_{ij}(t)+\epsilon_{ij}(t)$$
 
 \end{enumerate}

\end{theorem}

That is, the impact of a harmful shock $\epsilon_{ij}(t)>0$ on inequity $j$ is worsened by some other pre-existing inequity $k\neq j$, and the impact of a helpful shock $\epsilon_{ij}(t)<0$ on inequity $j$ is reduced by some other pre-existing inequity $k\neq j.$  The parameter characterizing the strength of the intersectoral synergy is $c_k>0$.  The interaction $c_k y_{ik}(s) \epsilon_{ij}(t)$ illustrates that the effect of the current period's shock $\epsilon_{ij}(t)$ on inequity $y_{ij}(t)$ in the current period depends on the value of another inequity $y_{ik}(s)$ in a previous time period.  $y_{ij}(t)$ also faces discount factor $\delta$ (``natural decay"), as any initial inequity would under the ideal asymptotically equitable process $(X(t))_{t \geq 0}$.  However, $S$ has transformed $(X(t))_{t \geq 0}$ into an amplified process $(Y(t))_{t \geq 0}$ that contains intersectoral synergy.  As a result, the transformed process $(Y(t))_{t \geq 0}$ will temporarily contain more inequity than $(X(t))_{t \geq 0}$ unless $\epsilon_{ij}(t)=0$.  A harmful shock $\epsilon_{ij}(t)>0$ will induce more inequity in $(Y(t))_{t \geq 0}$ than it would have in $(X(t))_{t \geq 0}$ and a helpful shock $\epsilon_{ij}(t)<0$ will reduce inequity by less in $(Y(t))_{t \geq 0}$ than it would have in $(X(t))_{t \geq 0}$.

\begin{restatable}{proposition}{Synergy}

Intersectoral synergy is discriminatory over an interval.

\end{restatable}

(Proof in Appendix)\\

Intersectoral synergy has important policy implications.  Let $y_{ik}(s)>0$ be a racial inequity in wealth for person $i$ at time $s$, and $y_{ij}(t)>0$ be a racial inequity in consumption at time $t$.  Suppose a harmful shock occurs in the housing market so that $\epsilon_{ij}(t)>0$. With public policy, we can disrupt the amplification mechanism by reducing $c_k$.  For example, we could insure poor households against sharp crashes in the housing market. 

We can also exploit the amplification mechanism by reducing $y_{ik}(s)$.  That is, we can reduce the size of the racial inequity in wealth with a wealth transfer.  In that case, a large $c_k$ would ensure that the wealth transfer retains the added effect of guarding the underprivileged group from systemic risk in the housing market.  In other words, the effects of the equity-focused intervention would amplify for the same reasons the effects of past injustices amplify systemic risk in the status quo.

In the long run, however, the importance of the enhanced vulnerability of a group to shocks will fade as long as the inequity $y_{ik}(t)$ that makes them vulnerable fades ``naturally" over time.  This changes if there is some other amplification mechanism or injustice that is continually enhancing the persistence of $y_{ik}(t)$, but intersectoral synergy alone is not discriminatory in the long run.

\begin{restatable}{proposition}{Synergylong}

Intersectoral synergy alone is not discriminatory in the long run.

\end{restatable}

(Proof in Appendix)\\

There are amplification mechanisms that can generate systemic discrimination in both the short and long run, described in the next subsections.  They are distinct from intersectoral spillover and intersectoral synergy in that they alone can endogenously sustain the conditions that allow them to function.

\subsection{Social Multipliers}

$S$ can induce $\textbf{Social Multipliers}$, so that inequities experienced by one individual percolate to other individuals within that person's social network.  Social multipliers constitute another body of “systemic features” through which inequities can amplify.  Suppose, for example that an African-American is denied hiring due to discrimination.  This impacts not only that individual, but all individuals connected to her network, who may have been relying on her for information relevant to that job or industry (Bolte et al., 2020).  In this case the amplification is about the percolation of an inequity throughout a population.

\begin{theorem}

\textbf{Social Multipliers} occur at time $t$ if for some $s<t$ and inequity $j$, there exists a collection of individuals $I$ and parameters $d_k > 0$, $k \in I$, s.t. for all $i \in I$,
 
$$y_{ij}(t)=\delta^{t-s} y_{ij}(s)+\sum_{k \neq i} d_k y_{kj}(s)+\epsilon_{ij}(t)$$

That is, if there is spillover across individuals.  
\end{theorem}

In this definition, the inequity $y_{ij}(t)$ faced by each person $i \in I$ is being influenced by the the presence of that inequity $y_{kj}(s)$ in other people $k \in \{k \neq i \}$ during an earlier time step $s<t$.  The parameters characterizing the nonzero strength of these ``peer effects" for person $i$ are $d_k \in \{k \neq i \}$.  Once again, the inequity $y_{ij}(t)$ faced by each person $i \in I$ is also influenced by its earlier value $y_{ij}(t)$, discounted at some ``natural" rate $\delta$.  This discounting would also have happened in the asymptotically equitable process $(X(t))_{t \geq 0}$, but in the amplified process $(Y(t))_{t \geq 0}$, there is also spillover across individuals.  As a result, there is, at least temporarily, more expected inequity at time t in the amplified process than there would have been in $(X(t))_{t \geq 0}$.

\begin{restatable}{proposition}{Socialshort}

A system that induces social multipliers is discriminatory over an interval.

\end{restatable}

(Proof in Appendix)\\

What do social multipliers suggest about combating systemic discrimination with public policy?  Let $y_{kj}(s)>0$ for all $k \in K$ be racial inequities in employment in a certain industry at time $s<t$ faced by a group of people in set $K$.  Suppose these lower the access to information held by other individuals contained in another set $I$, inducing employment inequities $y_{ij}(t)>0$ at time $t$ for all $i$ in $I$.  With public policy, we can disrupt the amplification mechanism by reducing $d_k$ for all $k \in K$.  For example, we can implement policies that desegregate referral networks, lessening the effect of a lack of Black Americans in an industry on other Black Americans' access to information relevant for employment.

We can also exploit the amplification mechanism by reducing $y_{kj}(s)>0$ for all $k \in K$.  For example, we can implement a temporary affirmative action policy that reduces inequities in employment in an industry.  In that case large parameters $d_k$ for all $K$ would help ensure the temporary affirmative action policy has persistent effects (Bolte et al., 2020).  In other words, once again, the effects of the equity-focused intervention would be amplified for the same reasons the effects of past injustices are amplified in the status-quo.

But the policy implications of social multipliers could be even more interesting.  Strong social multiplier effects can bolster the otherwise ``naturally" decaying inequities.  When social multipliers are sufficiently strong, injustices can have path dependent effects, and inequities can evolve with high sensitivity to initial conditions.  In other words, social multipliers can be discriminatory in the long run.  The fact that social multipliers can be discriminatory in the long run suggests inequities can cast long and permanent shadows in the absence of intervention.

\begin{restatable}{proposition}{Sociallong}

A system that induces sufficiently strong social multipliers is discriminatory in the long run.

\end{restatable}

(Proof in Appendix)\\

There is a powerful silver lining in this seemingly bleak case.  Especially at its strongest, systemic discrimination betrays a major vulnerability hiding in plain sight.  On one hand, rectifying the initial inequity below a tipping point would be essentially fruitless in the long run, like jumping slowly and falling back down under gravity.  On the other hand, a sufficiently strong intervention that rectifies the initial inequity could exploit the social multiplier, generating path-dependent, self-reinforcing effects-like jumping fast enough to achieve escape velocity.  Similarly, disrupting the strength of a social multiplier can induce a phase transition that fundamentally alters the long-term trajectory of the expected inequities, like reversing the direction of gravity.  The machinery of systemic discrimination can be used in its undoing.

   \begin{figure}[h]
 \centering
 \begin{tabular}{cc}
 \subf{\includegraphics[scale=.3]{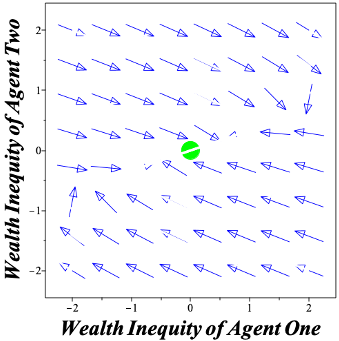}}
    {Weak Social Multipliers}
    &
     \subf{\includegraphics[scale=.3]{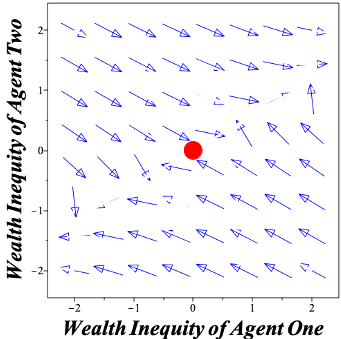}}
    {Strong Social Multipliers}
   \end{tabular}
   \caption{Strong Social Multipliers Make $S$ Discriminatory in the Long Run}
   \label{fig:Social1}
   \end{figure}

Consider the following policy example.  There is an empirical pattern in the literature that Black Americans in higher wealth ranks tend to have children whose wealth ranks drop more dramatically than do the children of their White counterparts (Pfeffer F.T., et al., 2015).  Among other explanations for this phenomenon, one possibility is social network effects (O'Brien, R.L., et al., 2012).  White Americans in upper wealth ranks have very different social networks than Black Americans in upper wealth ranks.  In particular, Black Americans in upper wealth ranks may have a network of disadvantaged people relying on them for financial support, whereas their equally wealthy White counterparts may even have the opposite-networks that can protect them from financial loss (Ager, P. et al., 2021).  As a result, segregated wealth transfer networks amplify/sustain wealth inequities (Meschede, T. et al., 2015).

This scenario is illustrated theoretically in Figure \ref{fig:Social1}.  For simplicity and clarity, the figure involves two agents, but the analogous principle holds for a population of $N>2$ people whose wealth inequities mutually influence one another.  This phase diagram is consistent with the system described at the end of the proof in the Appendix:

 \begin{align}
 E[y_{1j}(t)]=\delta E[y_{1j}(t-1)+d_1 E[y_{2j}(t-1)]\\
 E[y_{2j}(t)]=\delta E[y_{2j}(t-1)+d_2 E[y_{1j}(t-1)]
 \end{align}

In this example, $i=1,2$, $d_1$ captures the influence of Agent Two's wealth inequity on that of Agent One, and $d_2$ captures the influence of Agent One's wealth inequity on that of Agent Two.  The proof in the Appendix demonstrates that this system's equitable steady state of $(0,0)$ will only be stable if $\sqrt{d_1 d_2}<1-\delta$.  The social multiplier effects appear only as a product, suggesting that reciprocated social effects matter.  The effective strength of the social multiplier can be characterized by $\sqrt{d_1 d_2}$.

In the simulation, $\delta=0.6$ and $d_2=0.9$. In the first case, $d_1=0.1$, so the effective strength of the social multiplier is weak-that is, $\sqrt{d_1 d_2}<1-\delta$.  In this case the social multipliers only induce systemic discrimination in the short run-not in the long run.  Temporary interventions such as wealth transfers may have effects that persist longer due to social multipliers, but absent other amplification mechanisms or inustices, the inequities should eventually fade out over time regardless.  This is illustrated in Figure \ref{fig:Social1} on the left.  All trajectories eventually flow back to the steady state $(0,0)$ because it is stable, indicated by the green dot.   In the second case, $d_2=0.3$, so the effective strength of the social multiplier is strong-that is, $\sqrt{d_1 d_2}>1-\delta$.  Social multipliers are now strong enough to induce systemic discrimination in the long run, in that the effects of past injustices are sustained in a permanent, path-dependent fashion.  This is shown in Figure \ref{fig:Social1} on the right: the steady state of $(0,0)$ has become unstable, indicated by the red dot.  If the agents endured initial wealth inequities, the fact that their wealth levels are so closely interdependent, combined with the fact that their White counterparts' wealth are also interdependent, will sustain inequity for the foreseeable future.  A disruptive intervention that sufficiently decreases these interdependencies-for example, policies that desegregate wealth transfer networks-could induce a ``phase transition" that fundamentally alters the trajectories of wealth inequities, making the equitable steady state stable.  Alternatively, a dramatic wealth transfer that sufficiently decreases the wealth inequity beyond a tipping point will generate self-reinforcing effects.  Although the effects diverge due to the linearity of this model, this result is analogous to the more realistic case in nonlinear models in which trajectories flow to one of two stable steady states in the long run.

This mechanism has important implications for experiments on the fadeout or persistence of a reparations program for slavery and Jim crow.  Experiments that study the effects of reparations for individuals with social connections within otherwise broken communities, may underestimate the long term, self-replicating effects of a full-scale reparations program that improves the wealth of entire communities and social networks.    A similar argument has been made regarding universal pre-k (McMillon, 2024), and the argument applies across large-scale, equity-focused interventions.


The social multiplier achieves its power by allowing chain reactions across individuals.  But is not the only amplification mechanism that can generate systemic discrimination in the long run.  What follows is an exposition of a systemic feature that allows initial inequities to spark chain reactions across inequities.  This feature can be illustrated with a class of mathematical models called reinforcement processes.

\subsection{Reinforcement}

  $S$ can constitute a complex system of \textbf{reinforcement processes} for $(Y(t))_{t \geq 0}$, such that inequities are mutually reinforcing over time.  For example, there is strong feedback between wealth and neighborhood quality.  Wealth strongly influences neighborhood quality for several reasons including the fact that amenities that improve children's future productivity are costly.  Finally, neighborhood quality influences wealth through multiple channels, including determinants of income such as schooling quality, social capital, and crime (Chetty et al., 2016; Loury, G. 1977; Billings et al., 2019).  This cycle can continue even in the absence of continued direct discrimination.  As a consequence, wealth inequality can persist across generations. 

\begin{theorem}

Let $t \geq 0$ and suppose an injustice occurs at time $s<t$, inducing inequities $y_{ij}(s), y_{ik}(s)>0$.  A \textbf{Reinforcement Process} occurs for inequity $j$ at time $t$ if there exist $c_j, b_k > 0$, s.t. 

\begin{align}
y_{ij}(t) &= \delta y_{ij}(s)+ b_k y_{ik}(s)+\epsilon_{ij}(t)\\
y_{ik}(t) &= \delta y_{ik}(s)+ c_j y_{ij}(s)+\epsilon_{ik}(t)
\end{align}



That is, if an inequity feeds back into itself, whether directly or indirectly.

\end{theorem}

While intersectoral spillover allows the propagation of inequity across outcomes, reinforcement involves intersectoral spillover that eventually leads back to the initial inequitable outcome.  As in the above example, a key consequence of reinforcement processes is the persistence of inequities over time, since the inequities become their own causes and effects in a vicious cycle.  This perspective on systemic discrimination has never been formalized generally until now, but it is a major part of the argument regarding systemic discrimination by prominent sociologist Barbara Reskin (Reskin, B., 2012).

\begin{restatable}{proposition}{Reinforcementshort}

A system that induces reinforcement processes is discriminatory over an interval.

\end{restatable}

(Proof in Appendix)\\


What are the policy implications of reinforcement processes?  Let $y_{ij}(t-1)$ represent inequities in wealth at time $t-1$ for person $i$.  Suppose the wealth inequities at time $t-1$ impact neighborhood quality inequities at time $1$ through $c_j>0$, which reinforce wealth inequities at time $t+1$ through $b_k>0$.  With public policy, we can disrupt the reinforcement by reducing $c_j$.  For example, we can implement policies that allow poor families to live in high-quality neighborhoods, with place-based interventions or person-based housing voucher programs.  

Alternatively, we could exploit the amplification mechanism by reducing $y_{ij}(t-1)$ directly, possibly through baby bonds or a wealth-based reparations program.  In that case parameters $c_j, b_k$ would help ensure the wealth transfer has persistent effects.  In other words, once again, the effects of the equity-focused intervention would be amplified for the same reasons the effects of past injustices are amplified in the status-quo.

However, like social multipliers, reinforcement processes can bolster otherwise ``naturally" decaying inequities so much that the long-term trajectories of inequities are fundamentally altered.  Strong reinforcement processes can sustain one or many inequities within the same individual, and temporary interventions on one target inequity may be fruitless as the existence of other inequities eventually restore the target inequity anew.  When reinforcement processes are sufficiently strong, injustices can have path dependent effects, and inequities can evolve with high sensitivity to initial conditions.  In other words, reinforcement processes can be discriminatory in the long run.  If $S$ contains sufficiently strong reinforcement processes, injustices can cast long and permanent shadows in the absence of intervention. 

\begin{restatable}{proposition}{Reinforcementlong}

A system that induces sufficiently strong reinforcement processes is discriminatory in the long run.

\end{restatable}


\begin{proof}

Consider a function

$$S: (X(t))_{t \geq 0} \rightarrow (Y(t))_{t \geq 0}$$

acting on an asymptotically equitable random process $(X(t))_{t \geq 0}$ in such a way to produce reinforcement processes.  Then by definition, there exist $c_j, b_k > 0$, for some index $k \neq j$, s.t.

\begin{align}
y_{ij}(t+1) &= \delta y_{ij}(t)+ b_k y_{ik}(t)+\epsilon_{ij}(t+1)\\
y_{ik}(t+1) &= \delta y_{ik}(t)+ c_j y_{ij}(t)+\epsilon_{ik}(t+1)
\end{align}

where $E[\epsilon_{ij}(t+1)]=E[\epsilon_{ik}(t+1)]=0$.\\

Taking expectations results in the following system:

\begin{align}
E[y_{ij}(t+1)] &= \delta E[y_{ij}(t)]+ b_k E[y_{ik}(t)]\\
E[y_{ik}(t+1)] &= \delta E[y_{ik}(t)]+ c_j E[y_{ij}(t)]
\end{align}

If a steady-state exists, then $\lim_{t \rightarrow \infty} E[y_{ij}(t+1)]=\lim_{t \rightarrow \infty}  E[y_{ij}(t)]$ and $\lim_{t \rightarrow \infty} E[y_{ik}(t+1)]=\lim_{t \rightarrow \infty} E[y_{ik}(t)]$.  In that case the steady-state has the form

\begin{align}
Y &= \delta Y +b_k X\\
X &= \delta X+ c_j Y
\end{align}

This system has unique solution $(0,0)$.  Now consider the stability of the fixed point.  The Jacobian is 

\[ \mathbf{J} = \begin{pmatrix} \delta & b \\ c & \delta \end{pmatrix} \]

which leads to characteristic equation

\[ \begin{vmatrix} \delta - \lambda & b \\ c & \delta - \lambda \end{vmatrix} = (\delta - \lambda)^2 - bc \]

Therefore the eigenvalues are

\[ \lambda = \frac{2 \delta \pm \sqrt{4 \delta^2 - 4 \delta^2 + 4bc}}{2} \]

\[ \lambda = \frac{2 \delta \pm \sqrt{4bc}}{2} \]

\[ \lambda = \delta \pm \sqrt{bc} \]

\begin{figure}[h]
 \centering
 \begin{tabular}{ccc}
 \subf{\includegraphics[scale=.3]{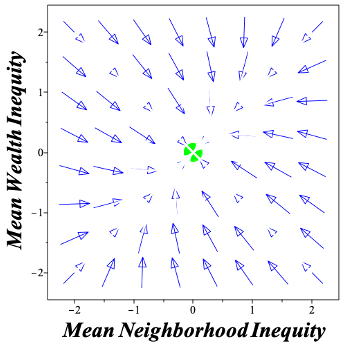}}
    {Weak Reinforcement}
    &
     \subf{\includegraphics[scale=.3]{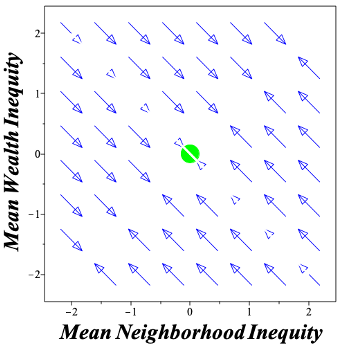}}
    {Phase Transition}
     &
     \subf{\includegraphics[scale=.3]{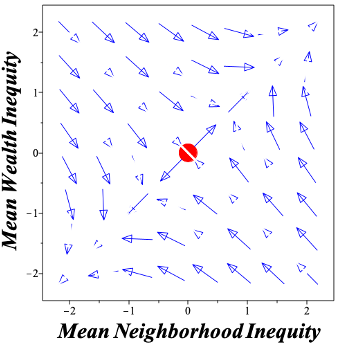}}
    {Strong Reinforcement}
   \end{tabular}
   \caption{Strong Reinforcement Makes S Discriminatory in the Long Run}
   \label{fig:Reinforcement1}
   \end{figure}

The system is stable if\\

\(|\delta + \sqrt{bc}| < 1\) and \(|\delta - \sqrt{bc}| < 1\)\\

Since we have assumed $|\delta|<1$ and $b,c>0$, the condition simplifies to\\

\(|\delta + \sqrt{bc}| < 1\).\\

Therefore, for fixed $\delta$, as $bc$ increases, the system can undergo a phase transition in which the (0,0) steady state goes from stable to unstable.  For sufficiently large $bc$, any initial inequity for $y_{ik}(0)$ or $y_{ij}(0)$ will diverge.  Again, this divergence is only due to the linearity of this model; this result is analogous to the more realistic case in nonlinear models in which the unstable $(0,0)$ state separates two stable steady states in the long run.  The phase transition threshold for $bc$ becomes smaller as $\delta$ increases.  Qualitatively, the fact that $bc$ appears as a product suggests it is the mutual reinforcement of two spillovers, and not each spillover separately, that matters.  When the steady state is stable, the convergence rate is slower for larger values of $\delta+bc$, so strong reinforcement induces bottlenecks as well.\\

   \begin{figure}[h]
 \centering
 \begin{tabular}{cc}
 \subf{\includegraphics[scale=.3]{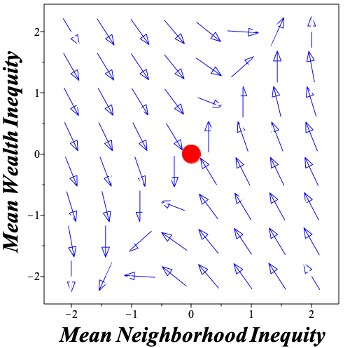}}
    {Strong Reinforcement, $b>c$}
    &
     \subf{\includegraphics[scale=.3]{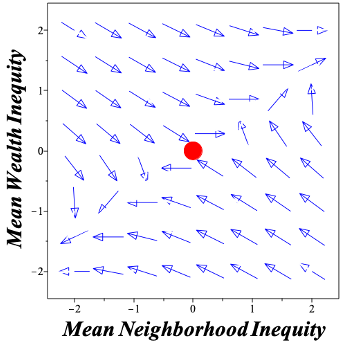}}
    {Strong Reinforcement, $c>b$}
   \end{tabular}
   \caption{Strength of c vs. b Determines Sensitivity to Initial Conditions}
   \label{fig:Reinforcement2}
   \end{figure}

Suppose there is reinforcement between inequities in neighborhood quality ($y_{ik}(t)$) and inequities in wealth ($y_{ij}(t)$).  Above are figures depicting the dynamics of the continuous time version of this system with phase portraits.  Figure \ref{fig:Reinforcement1} shows how for fixed $\delta=0.5$, increasing the strength of reinforcement from $bc=0.2$ to $bc=0.5$ to $bc=0.8$ leads to different dynamics.  In the first case, the equitable steady state (0, 0) is stable, and any inequities will eventually dissipate.  In the second case, a phase transition occurs such that initial inequities proceed to the $45$ degree line and remain thereafter.  In this case there are infinitely many stable steady states, each determined by the initial inequities.  In the final case, the equitable steady state (0, 0) is unstable. Initial inequities diverge, and, were the initial conditions reversed, they would diverge in the opposite direction.  In this case $S$ is discriminatory in the long run since the expected inequities diverge relative to what they would be in the steady state ($0, 0$).   Figure \ref{fig:Reinforcement2} considers only the case in which the steady state is unstable.  It demonstrates that the direction of the sensitivity to initial conditions depends on the relative strength of the reinforcement parameters.  In this example, if the wealth inequity is more sensitive to neighborhood quality inequity, than neighborhood quality inequity is to wealth inequity ($b>c$), then the direction of the initial inequity in neighborhood quality matters the most for what happens in the long run.  The reverse is true when $c>b$.  It can be shown that similar intuition holds for $M>2$ mutually reinforcing inequities.

\end{proof}

The policy implications of this case cannot be overstated.  If reinforcement processes induce systemic discrimination in the long run, then potentially, not just one, but several inequities can be mutually sustained in a path dependent fashion.  This is consistent with a world in which ongoing generations descended from survivors of slavery and Jim Crow would face a wide range of mutually reinforcing inequities across generations in wealth, health, education, income, neighborhood quality, and criminal justice outcomes.  In such a world, the effects of temporary interventions that induce successful but modest reductions in health inequities (e.g., a health education intervention) may tend to fade out in the long run as other inequities allow the health inequities to resurface (e.g., wealth inequities that alter access to nutrition).  The long run effects of injustices, sustained through a mutually reinforcing, interrelated inequities, can frustrate public policy.  This concern has been fought by academic voices and approaches across multiple fields, from the sociological notion of ``Uber discrimination" (Reskin, 2012), to ecological systems theory in education (Patton et al., 2012) and urban planning (Schell et al., 2020), to fundamental cause theory in public health (Phelan, J.C. et al., 2010; Phelan, J.C. et al., 2015).

However, once again, the greatest strength of systemic discrimination is actually its Achilles' Heel hiding in plain sight.  On one hand, rectifying the initial inequities below a tipping point would be fruitless in the long run, like traveling towards a star that is simply too far to reach in a human lifetime.  Yet, a strong, courageous intervention that rectifies the appropriate combination of inequities could exploit the physics of the reinforcement system for what would essentially be a ``quantum leap" in equity.  Beyond a threshold, the reinforcement processes would serve as a booster engine for the effects of the equity-focused intervention.  Similarly, disrupting the strength of the feedback loops can induce a phase transition that fundamentally alters the long-term trajectory of the expected inequities, like folding spacetime to shorten the distance to a star.  This would all but destroy the sensitivity of long-run inequities to the unjust initial conditions.  The physics of systemic discrimination can be harnessed for its undoing.  




\section{Discussion: Limitations and Future Directions}

The main contribution of this paper is the formalization of a non-exhaustive taxonomy of mechanisms that amplify inequities in various ways.  This section highlights the limitations of the formal model itself, including implications for measurement, testable predictions, and policy.  

\subsection{Theory and Measurement, and Policy}

The model presented in this paper is intentionally simple, to clearly illustrate how amplification mechanisms contribute to inequities.  There are many other functional forms and modeling approaches I could have taken, but this was the simplest model that makes the key point of the paper.  However, the simplicity has led to several limitations.  

First, the model assumes the disadvantaged group does not change its group membership.  This is a safe assumption for the history of Black Americans.  It is not a safe assumption for Mexican Americans and ethnic Jews.  The definition of ``White" has changed over time in America, subsuming some groups such as ethnic Jews, and removing others such as Mexican Americans.  It could easily be extended to account for group membership, but such a complexity could distract from the central point of the paper.

Second, the matrices $X(t)$ and $Y(t)$ only keep track of the disadvantaged group relative to the advantaged group.  In principle the model does implicitly include the advantaged group's means, but does not keep track of individuals within the advantaged group.  However, the model can easily be altered to keep track of members of the advantaged group relative to group means of the disadvantaged group.  The main implications of the model would not change.  However, it could orient the kinds of examples the reader considers.  For example, social multipliers may amplify inequity through the spread of injustice across a subset of people in the advantaged group, rather than the spread of the effects of injustice among the disadvantaged.  Derek Chauvin murdered George Floyd, but his peers watched it happen.  Peer effects can allow the total contribution of the police force to violence against Black Americans to be more than the sum of the individual contributions were the police acting in isolation.

Third, the ``asymptotically equitable" random process $(X(t))_{t \geq 0}$ was constructed as a Markov process with decay factor $\delta \in (0,1)$.    If $\delta>1$, then initial inequities will never dissipate.  This is essentially what we might expect if White Americans simply invested their wealth advantages as compound interest in every generation following slavery.  But Section $3.4$ distinguishes between the ``natural decay" of an inequity through $\delta \in (0,1)$, and direct or indirect reinforcement.  In the case of compound interest, the reinforcement is direct, so $j=k$ and $c_j=b_k>0$.  Therefore the restriction $\delta \in (0,1)$ is without loss of generality.

Fourth, I imposed linear specifications for social multipliers and reinforcement processes.  This was intended for pedagogical simplicity.  Both of these amplification mechanisms can generate systemic discrimination in the long run, but this appears in the linear models as either a knife-edge case, or as expected inequities diverging over time.  While this is empirically unrealistic, something much more realistic and analogous can happen in nonlinear versions of the same models.  When the strength of social multipliers or reinforcement is sufficiently strong in nonlinear models, multiple steady states can exist.  Initial inequities can lead to converge to a finite but unfavorable steady state, which the system would not have converged to under different initial conditions.  I chose to keep linear specifications because divergence in the linear case is still indicative of more realistic path dependence in nonlinear cases.

Fifth, some notions of systemic discrimination argue that a dominant group can intentionally alter social and political systems to retain power.  $S$ is a very flexible function that can embed political economy models in which a dominant group can respond strategically to policy interventions.  This paper is intentionally agnostic about the nature of $S$ beyond the abstract ways it can amplify inequities.  Whether this occurs in part due to the endogenous response of a dominant group, intentional or not, is beyond the scope of this paper.

Finally, the model is constructed to consider the ramifications of temporary injustices and temporary interventions.  It is flexible enough to account for ongoing, endogenously generated discrimination, but it does not consider exogenously imposed permanent injustices.  If there is some baseline constant level of animus automatically generated by society, then I would either need to leave the definition of $(X(t))_{t \geq 0}$ as it is and accept that such factors also count as ``systemic," or I would need to redefine $(X(t))_{t \geq 0}$ to be equitable in the long run up to a constant $c$ that characterizes the unavoidable baseline inequity.  Neither case rules out the importance of focusing on amplification mechanisms.  In fact, the presence of such ongoing exogenously generated injustices makes fighting more strategically with amplifiers all the more important.

This model generates testable predictions, since the presence of these amplification mechanisms can be tested with experimental or quasi-experimental methods.  Intersectoral spillover is evidenced by a causal effect of one outcome on another through standard methods.  It follows from the fact that one outcome impacts another, that an injustice inducing inequities in one will spill over into another.  Intersectoral synergy is more subtle.  One needs two sources of random variation.  It is evidenced either when the negative impact of a harmful shock is worsened in those induced to retain a lower value of another variable, or when the positive impact of a helpful shock is heightened in those induced to retain a higher value of another variable.  An example of the latter case is Johnson's and Jackson's approach to identifying dynamic complementarity (Johnson and Jackson, 2019):  those exposed to head start benefited more from desegregation later in life.  There is also an extensive literature on the identification of social multipliers interaction effects, in discrete choice models and networks (Glaeser et. al, 2003; Blume et. al, 2005; Durlauf et. al, 2010).

The identification of reinforcement processes in social science contexts can be challenging and differs across fields.  A key concern is time-varying confounding, since outcomes and regressors are mutually reinforcing over time.  Identification is possible with sequential exogeneity (past error terms can be correlated with later regressors, but all past regressors must be uncorrelated with future error terms).  Sociologists Paul Allison and coauthors (2017) discuss the estimation of dynamic panel data models under sequential exogeneity.  Wodtke (2018) also proposes a method to handle time-varying confounding.  Many macroeconomic models include reinforcement processes, which are generally estimated structurally.  The evidence of the processes' existence lies with a combination of calibration, exogenous variation, and in the extent to which their predictions are consistent with empirically observed phenomena.  

Finally, the taxonomy included policy examples under each amplification mechanism.  Although an extensive discussion of the pros and cons of specific policies is outside the scope of this research, the paper does illustrate how the model can be used to think through how to harness public policy to fight systemic discrimination.  The main purpose of the policy discussions was to illustrate what role the intended consequences of these kinds of policies might play in the calculus of systemic discrimination.  Specifically, each example involved a policy that directly disrupts an amplifier of inequity, and a policy that fights the inequity itself and exploits the amplifier in the process.  Collectively, they demonstrate how to use the machinery of systemic discrimination against itself across a wide range of contexts.

Note, however, that each of these policies have pros and cons.  For example, place-based interventions that improve neighborhoods can lead to gentrification and displacement, and housing voucher programs may eventually lead to White flight.  Despite that, an intervention that successfully disrupts the reinforcement between wealth and neighborhood quality will necessarily reduce the sensitivity of the process to its unjust initial conditions; and an intervention that successfully exploits the reinforcement with an exogenous wealth transfer will necessarily benefit from persistent effects.  Similarly, there is conflicting evidence on the effects of affirmative action (Fryer, R. and Loury, G., 2005; Kurtulus, F., 2016), since it is feasible that it can contribute to statistical discrimination.  Despite this, the heavy reliance on referrals can contribute to the possibility of persistent effects for a temporary affirmative action shock, and the disruption of segregated referral networks can reduce the sensitivity of labor market inequities to unjust initial conditions.

\section{Conclusion}

\textit{``Everything we see is a shadow cast by that which we do not see"} \vspace{0.3cm}

-\textbf{MLK, Jr.}\\

We cannot see gravity, but we can see its effects.  We can prevent it from harming us, or even exploit it to our advantage--but not before we understand it well enough to formalize and test.  Analogously, the large and persistent inequities we see are not simply due to injustices in a vacuum.  They are also shadows cast by insidious, seemingly benign processes that amplify their initial effects.  This perspective reveals that these insidious amplification mechanisms can either be directly disrupted, or exploited, to ensure that the effects of equity-focused interventions are amplified instead.  This paper develops a technical, testable, and policy-oriented understanding of what makes systemic discrimination ``systemic."  It reveals how the machinery of systemic discrimination can be used against itself.  Finally, because the amplification mechanisms themselves connect well-understood concepts in economics, the gap this paper fills could lead to more rapid, strategic, and coordinated strides in the race towards justice.





    \newpage

\section{Appendix}

\subsection{Proofs}


\Spillover*

\begin{proof}

Suppose the existence of inequities $y_{ik}(s)>0$ and $y_{ij}(s)=x_{ij}(s)$. Consider a function

$$S: (X(t))_{t \geq 0} \rightarrow (Y(t))_{t \geq 0}$$

acting on an asymptotically equitable random process $(X(t))_{t \geq 0}$ in such a way to produce intersectoral spillover.  Then by definition, for some $s<t$ and inequity $j$, there exists $b_k > 0, k \neq j$, s.t.

$y_{ij}(t)=\delta^{t-s} y_{ij}(s)+b_k y_{ik}(s)+\epsilon_{ij}(t)$.  Therefore

 \begin{align}
y_{ij}(t)=\delta^{t-s} y_{ij}(s)+b_k y_{ik}(s)+\epsilon_{ij}(t)\\
=\delta^{t-s} x_{ij}(s)+b_k y_{ik}(s)+\epsilon_{ij}(t)\\
>\delta^{t-s} x_{ij}(s)+\epsilon_{ij}(t)\\
=x_{ij}(t)
  \end{align}
  
 It follows that $E[y_{ij}(t)]>E[x_{ij}(t)]$. Hence $S$ is discriminatory over the interval $(s,t)$.
  
\end{proof}

\Spilloverlong*

\begin{proof}

Suppose the existence of inequities $y_{ik}(t-1)>0, y_{ij}(t-1)=x_{ij}(t-1)$. Consider a function

$$S: (X(t))_{t \geq 0} \rightarrow (Y(t))_{t \geq 0}$$

acting on an asymptotically equitable random process $(X(t))_{t \geq 0}$ in such a way to produce intersectoral spillover.  Then by definition, for some inequity $j$, there exists $b_k > 0, k \neq j$, s.t.

$$y_{ij}(t)=\delta y_{ij}(t-1)+b_k y_{ik}(t-1)+\epsilon_{ij}(t)$$  

In contrast, for the asymptotically equitable random process $(X(t))_{t \geq 0}$, 

$$x_{ij}(t)=\delta y_{ij}(t-1)+\epsilon_{ij}(t)$$  

Taking expectations results in the system

 \begin{align}
E[y_{ij}(t)] &= \delta E[y_{ij}(t-1)]+b_k E[y_{ik}(t-1)]\\
E[x_{ij}(t)] &= \delta x_{ij}(t-1)\\
  \end{align}
  
If a steady-state exists, then $\lim_{t \rightarrow \infty} E[y_{ij}(t)]=\lim_{t \rightarrow \infty}  E[y_{ij}(t-1)]$ and $\lim_{t \rightarrow \infty} E[x_{ij}(t)]=\lim_{t \rightarrow \infty} E[x_{ij}(t-1)]$.  In that case the steady-state has the form

\begin{align}
Y &= \delta Y +b_k X\\
X &= \delta X
\end{align}

which implies a steady state of $(0,0)$.  The Jacobian is:

\[
J = \begin{bmatrix} \delta & b_k \\ 0 & \delta \end{bmatrix}
\]

which leads to characteristic equation\\

\[
\text{det}(J - \lambda I) = \text{det} \left( \begin{bmatrix} \delta-\lambda & b_k \\ 0 & \delta-\lambda \end{bmatrix} \right)
\]
\[
= (\delta-\lambda)^2
\]

Therefore the eigenvalues are $\lambda=\delta$, multiplicity of $2$.  Since $0<\delta<1$ by assumption, the steady state $(0,0)$ is stable.  Convergence depends only on $\delta$ and is slower for larger values of $\delta$.\\

Altogether, $$ \lim_{t \rightarrow \infty} E[y_{ij}(t)]=0=\lim_{t \rightarrow \infty}  E[x_{ij}(t)] $$

$$\lim_{t \rightarrow \infty}  E[y_{ik}(t)]=0=\lim_{t \rightarrow \infty} E[x_{ik}(t)]$$

and $S$ is not discriminatory in the long run.  Essentially, Intersectoral Spillover alone generates short-run inequity, but is not enough on its own to generate long-run inequity.\\  

\end{proof}

\Synergy*

\begin{proof}


 
 

 
 



Suppose without loss of generality the existence of inequities $y_{ik}(0)>0$ and $y_{ij}(0)=x_{ij}(0)$. Consider a function

$$S: (X(t))_{t \geq 0} \rightarrow (Y(t))_{t \geq 0}$$

acting on an asymptotically equitable random process $(X(t))_{t \geq 0}$ in such a way to produce intersectoral synergy.  Then by definition, 
 
 \begin{enumerate}
 
 \item If $\epsilon_{ij}(t)>0$, there exists $c_k>0$, $k \neq j$, s.t.
$$y_{ij}(t)=\delta y_{ij}(t-1)+c_k y_{ik}(t-1) \epsilon_{ij}(t)+\epsilon_{ij}(t)$$

\item If $\epsilon_{ij}(t)<0$, there exists $c_k<0$, $k \neq j$, s.t.
$$y_{ij}(t)=\delta y_{ij}(t-1)+c_k y_{ik}(t-1) \epsilon_{ij}(t)+\epsilon_{ij}(t)$$
 
 \end{enumerate}
 
 Recall that $$y_{ik}(t)=\delta y_{ik}(t-1)+\epsilon_{ik}(t)$$ 
 
 which implies that 
 
 $$y_{ik}(t)=\delta^t y_{ik}(0)+\epsilon_{ik}(t)$$ 
 
 The way the process $y_{ij}(t)$ evolves depends on the sign of $\epsilon_{ij}(t)$.  Solve this by conditioning, recognizing that $\epsilon_{ij}(t)$ is independent of $\epsilon_{ik}(t-1)$ and $y_{ik}(t-1)$:
 
 $$E[y_{ij}(t)|\epsilon_{ij} \geq 0] = \delta E[y_{ij}(t-1)|\epsilon_{ij}(t) \geq 0]+c_k E[y_{ik}(t-1) \epsilon_{ij}(t)|\epsilon_{ij}(t) \geq  0]+E[\epsilon_{ij}(t)|\epsilon_{ij}(t) \geq  0]$$
 
 $$= \delta E[y_{ij}(t-1)|\epsilon_{ij}(t) \geq 0]+c_k E[y_{ik}(t-1)] E[\epsilon_{ij}(t)|\epsilon_{ij}(t)>0]+E[\epsilon_{ij}(t)|\epsilon_{ij}(t) \geq  0]$$
 
 $$= \delta E[y_{ij}(t-1)|\epsilon_{ij}(t) \geq 0]+c_k \left( \delta^t y_{ik}(0) \right) E[\epsilon_{ij}(t)|\epsilon_{ij}(t) \geq 0]+E[\epsilon_{ij}(t)|\epsilon_{ij}(t) \geq 0]$$
 
 Define $X(0):= y_{ik}(0)$, and $e_j := E[\epsilon_{ij}(t)|\epsilon_{ij}(t) \geq 0]$ Note that $e_j$ may not equal 0.  Also define $Y(t):=E[y_{ij}(t)|\epsilon_{ij}(t) \geq 0]$ for $t \geq 1$, with $Y(0)=y_{ij}(0)$.  Then our analysis leads us to the recursive process
 
 $$Y(t)=\delta Y(t-1)+c_k \delta^t X(0) e_j+e_j$$
 
 which has initial condition $Y(1)=\delta Y(0)+c_k X(0) e_j+e_j$, and for $t \geq 2$, the explicit form
 
 $$Y(t)= \delta^t Y(0)+t c_k \delta^t X(0) e_j + e_j \sum_{l=1} \delta^{t-l}$$
 
 Substituting,\\
 
 $$E[y_{ij}(t)|\epsilon_{ij}(t) \geq 0]=\delta^t y_{ij}(0)+t  \delta ^t c_k y_{ik}(0) e_j+e_j$$
Comparing this to how inequity $j$ would evolve for person $i$ under the asymptotically equitable random process $(X(t))_{t \geq 0}$:

\begin{align}
E[y_{ij}(t)|\epsilon_{ij}(t) \geq 0]= \delta^t y_{ij}(0)+t  \delta ^t c_k y_{ik}(0) e_j+e_j\\
= \delta^t x_{ij}(0)+t  \delta ^t c_k y_{ik}(0) e_j+e_j\\
\geq \delta^t x_{ij}(0)+e_j\\
=E[x_{ij}(t)|\epsilon_{ij}(t) \geq 0]
\end{align}

Here the inequality stems from the fact that $c_k e_j \geq 0$.\\  

A symmetric analysis applies for the case in which $\epsilon_{ij}(t)<0$.  Define $v_j:=E[\epsilon_{ij}|\epsilon_{ij}(t)<0]$.  Comparing this to how inequity $j$ would evolve for person $i$ under the asymptotically equitable random process $(X(t))_{t \geq 0}$:

\begin{align}
E[y_{ij}(t)|\epsilon_{ij}(t) < 0]= \delta^t y_{ij}(0)+t  \delta ^t c_k y_{ik}(0) v_j+v_j\\
= \delta^t x_{ij}(0)+t  \delta ^t c_k y_{ik}(0) v_j+v_j\\
> \delta^t x_{ij}(0)+v_j\\
=E[x_{ij}(t)|\epsilon_{ij}(t)<0]
\end{align}

The strict inequality stems from the fact that $c_k v_j > 0$.  Finally,

\begin{align}
E[y_{ij}(t)]=E[y_{ij}(t)|\epsilon_{ij}(t) \geq 0] Prob(\epsilon_{ij}(t) \geq 0)+E[y_{ij}(t)|\epsilon_{ij}(t) < 0] Prob(\epsilon_{ij}(t) < 0) \\
\geq E[x_{ij}(t)|\epsilon_{ij}(t) \geq 0] Prob(\epsilon_{ij}(t) \geq 0)+E[y_{ij}(t)|\epsilon_{ij}(t) < 0] Prob(\epsilon_{ij}(t) < 0) \\
> E[x_{ij}(t)|\epsilon_{ij}(t) \geq 0] Prob(\epsilon_{ij}(t) \geq 0)+E[x_{ij}(t)|\epsilon_{ij}(t) < 0] Prob(\epsilon_{ij}(t) < 0) \\
=E[x_{ij}(t)]
\end{align}

Therefore intersectoral synergy is discriminatory on the interval $(0,t)$.  Essentially, although $E[\epsilon_{ij}(t)]=0$, any nonzero value of $\epsilon_{ij}(t)$ will have induced more inequity at time $t$ under a transformed process $(Y(t))_{t \geq 0}$ operating through intersectoral synergy than it would have under the asymptotically equitable process $(X(t))_{t \geq 0}$.  So, outside of the trivial case in which $\epsilon_{ij}(t)$ is always exactly equal to zero, intersectoral synergy is discriminatory on an interval.
\end{proof}

\Synergylong* 

\begin{proof}

Suppose without loss of generality the existence of inequities $y_{ik}(0)>0$ and $y_{ij}(0)=x_{ij}(0)$. Consider a function

$$S: (X(t))_{t \geq 0} \rightarrow (Y(t))_{t \geq 0}$$

acting on an asymptotically equitable random process $(X(t))_{t \geq 0}$ in such a way to produce intersectoral synergy.  Then by definition, 
 
 \begin{enumerate}
 
 \item If $\epsilon_{ij}(t)>0$, there exists $c_k>0$, $k \neq j$, s.t.
$$y_{ij}(t)=\delta y_{ij}(t-1)+c_k y_{ik}(t-1) \epsilon_{ij}(t)+\epsilon_{ij}(t)$$

\item If $\epsilon_{ij}(t)<0$, there exists $c_k<0$, $k \neq j$, s.t.
$$y_{ij}(t)=\delta y_{ij}(t-1)+c_k y_{ik}(t-1) \epsilon_{ij}(t)+\epsilon_{ij}(t)$$
 
 \end{enumerate}
 
 Following the analysis from the previous proof, notice that
 
 $$\lim_{t \rightarrow \infty} E[y_{ij}(t)|\epsilon_{ij}(t) \geq 0]= \lim_{t \rightarrow \infty}  \left( \delta^t y_{ij}(0)+t  \delta ^t c_k y_{ik}(0) e_j+e_j \right)=e_j=E[\epsilon_{ij}(t)|\epsilon_{ij}(t) \geq 0]$$
 
 and 
 
  $$\lim_{t \rightarrow \infty} E[y_{ij}(t)|\epsilon_{ij}(t) \geq 0]= \lim_{t \rightarrow \infty}  \left( \delta^t y_{ij}(0)+t  \delta ^t c_k y_{ik}(0) v_j+v_j \right)=v_j=E[\epsilon_{ij}(t)|\epsilon_{ij}(t) < 0]$$
  
Notice the term containing the strength $c_k$ of intersectoral synergy fades as $t$ approaches infinity.  Altogether,

 \begin{align}
 \lim_{t \rightarrow \infty} E[y_{ij}(t)]=\lim_{t \rightarrow \infty} E[y_{ij}(t)|\epsilon_{ij}(t) \geq 0] Prob(\epsilon_{ij}(t) \geq 0)+\lim_{t \rightarrow \infty} E[y_{ij}(t)|\epsilon_{ij}(t) < 0] Prob(\epsilon_{ij}(t) < 0)\\
 =E[\epsilon_{ij}(t)|\epsilon_{ij}(t) \geq 0] Prob(\epsilon_{ij}(t) \geq 0)+E[\epsilon_{ij}(t)|\epsilon_{ij}(t) < 0] Prob(\epsilon_{ij}(t) < 0)\\
 =E[\epsilon_{ij}(t)]=0=\lim_{t \rightarrow \infty} E[x_{ij}(t)]
 \end{align}

 Since $\lim_{t \rightarrow \infty} E[y_{ij}(t)]=\lim_{t \rightarrow \infty} E[x_{ij}(t)]$, intersectoral synergy is not discriminatory in the long run.  Essentially, although it is discriminatory in finite intervals due to the way $y_{ij}(0)$ interacts with $\epsilon_{ij}(t)$ through $c_k$, the influence of $c_k$ fades away as $t$ approaches infinity.

\end{proof}

\Socialshort*

\begin{proof}
Let $j$ index an arbitrary inequity.  Consider a function

$$S: (X(t))_{t \geq 0} \rightarrow (Y(t))_{t \geq 0}$$

acting on an asymptotically equitable random process $(X(t))_{t \geq 0}$ in such a way to produce social multipliers.   Then by definition, there exists a collection of individuals $I$ and parameters $d_k > 0$, $k \in I$, s.t. for all $i \in I$,
 
$$y_{ij}(t)=\delta^{t-s} y_{ij}(s)+\sum_{k \neq i} d_k y_{kj}(s)+\epsilon_{ij}(t)$$

Suppose the existence of inequities $y_{kj}(s)>0$ for all $k \in I$ and $y_{ij}(s)=x_{ij}(s)$. 

It follows that

 \begin{align}
y_{ij}(t)=\delta^{t-s} y_{ij}(s)+\sum_{k \neq i} d_k y_{kj}(s)+\epsilon_{ij}(t)\\
=\delta^{t-s} x_{ij}(s)+\sum_{k \neq i} d_k y_{kj}(s)+\epsilon_{ij}(t)\\
>\delta^{t-s} x_{ij}(s)+\epsilon_{ij}(t)\\
=x_{ij}(t)
  \end{align}

 Therefore $E[y_{ij}(t)]>E[x_{ij}(t)]$. Hence $S$ is discriminatory over the interval $(s,t)$.
  
  \end{proof}

\Sociallong*

\begin{proof}
Let $j$ index an arbitrary inequity.  Consider a function

$$S: (X(t))_{t \geq 0} \rightarrow (Y(t))_{t \geq 0}$$

acting on an asymptotically equitable random process $(X(t))_{t \geq 0}$ in such a way to produce social multipliers.  Then by definition, there exists a collection of individuals $I$ and parameters $d_k > 0$, $k \in I$, s.t. for all $i \in I$,
 
$$y_{ij}(t)=\delta^{t-s} y_{ij}(s)+\sum_{k \neq i} d_k y_{kj}(s)+\epsilon_{ij}(t)$$

Suppose without loss of generality the existence of common inequities $x_{ij}(0)=y_{ij}(0)=x_0>0$ for all $i \in I$.  Suppose also without loss of generality a common social effect $d$ from each individual $k \in \{k \neq i \}$, so that $d_k=d>0$ for all $k \in \{k \neq i \}.$\\

It follows that

$$y_{ij}(t)=x_0(\delta+d)^t+\sum_{s=1}^{t-1} d^{t-s} \epsilon_{ij}(s)+\sum_{s=1}^t \delta^{t-s} \epsilon_{ij}(s)$$

Taking limiting expectations, we compute $\lim_{t \rightarrow \infty} E[x_{ij}(t)]$:

\begin{align}
 =\lim_{t \rightarrow \infty} E[\delta^{t} x_0+ \sum_{k=1}^t \delta^{t-k} \epsilon_{ij}(k)]\\
 =\lim_{t \rightarrow \infty}  \delta^t x_0\\
 =0
 \end{align}
 
 Similarly, we compute $\lim_{t \rightarrow \infty} E[y_{ij}(t)]$:
 
 \begin{align}
=\lim_{t \rightarrow \infty} x_0(\delta+d)^t+\sum_{s=1}^{t-1} d^{t-s} E[\epsilon_{ij}(s)]+\sum_{s=1}^t \delta^{t-s} E[\epsilon_{ij}(s)]\\
 =\lim_{t \rightarrow \infty} x_0(\delta+d)^t 
\end{align}

Consider three cases.  First, suppose $d<1-\delta$.  Then 

$$\lim_{t \rightarrow \infty} E[y_{ij}(t)]=\lim_{t \rightarrow \infty} x_0(\delta+d)^t=0=\lim_{t \rightarrow \infty} E[x_{ij}(t)]$$
    
 In this case, the social multiplier parameter $d$ is too weak to generate systemic discrimination in the long run.  There is a unique stable steady state in which the long run expected inequity value is $0$.   Note, however, that $E[y_{ij}(t)]=x_0(\delta+d)^t$.  The ratio of consecutive terms of the sequence is equivalent to 
 
 $$\frac{(\delta+d)^{t+1}}{(\delta+d)^t}=\delta+d.$$
 
 The ratio is closer to $1$ as $d$ increases towards $1-\delta$, so stronger social multipliers slow down the convergence rate of the sequence, inducing greater and greater persistence.\\
 
 Next, suppose $d=1-\delta$.  Then
 
 $$\lim_{t \rightarrow \infty} E[y_{ij}(t)]=\lim_{t \rightarrow \infty} x_0(\delta+d)^t=x_0>0=\lim_{t \rightarrow \infty} E[x_{ij}(t)]$$
 
 In this case, the social multiplier parameter $d$ is just strong enough to generate systemic discrimination in the long run.  There are also infinitely many steady states, which are completely determined by the initial inequity value $x_0$.\\
 
 As the social multiplier value increases above the threshold value of $1-\delta$, the system experiences a phase transition in which inequities diverge in the long run.  As the discount factor $\delta$ increases, divergence can be achieved with smaller and smaller values of $d$.  As long as $x_0 \neq 0$, there is no steady state, and the social multipliers are so strong that any initial inequity amplifies in expectation without bound across the population.  Clearly, in that case, since $E[y_{ij}(t)]$ diverges to infinity, it exceeds $0=\lim_{t \rightarrow \infty} E[x_{ij}(t)$, and the system is discriminatory in the long run.\\ 
 
 There are many specifications involving social multipliers, interactions and contagion which are nonlinear and generate multiple steady states outside of a knife-edge case (McMillon, D., 2024; McMillon, D. et al., 2014; Brock et al., 2004).  The only reason multiple steady states are only possible for a single value of $d$ here is that I have proposed a linear specification for pedagogical simplicity and clarity.  In this model, the $\delta+d=1$ case of this is analogous to a phase transition between a single stable and multiple stable steady states in nonlinear models. The divergent case is analogous to cases with multiple stable steady states in nonlinear models.  In those cases, there is a steady state that has become unstable to the right of a threshold, and trajectories travel to some newly stable steady state instead of diverging, which happens in linear models.  The qualitative policy implications are the same-that disruptive interventions that change the value of the threshold could have powerful long run effects, and that sufficiently strong interventions that rectify the initial inequities can exploit the value of the threshold for long run effects.\\
 
 The simplifications are without loss of generality, and similar intuition holds when they are relaxed.  In the case in which there are only two people, consider the system\\
 
 \begin{align}
 E[y_{1j}(t)]=\delta E[y_{1j}(t-1)+d_1 E[y_{2j}(t-1)]\\
 E[y_{2j}(t)]=\delta E[y_{2j}(t-1)+d_2 E[y_{1j}(t-1)]
 \end{align}
 
 such that $d_1 \neq d_2$.  That is, I have relaxed the assumption that the social multiplier effect is the same across all people.  This system has Jacobian
 
 \[
J = \begin{pmatrix}
\delta & d_1 \\
d_2 & \delta
\end{pmatrix}
\]
 
 with characteristic equation
 
 \[
(\delta - \lambda)^2 - d_1d_2 = 0
\]

and eigenvalues

\[
\lambda = \delta \pm \sqrt{d_1d_2}
\]

Stability at (0,0) is guaranteed when $\sqrt(d_1d_2)<1-\delta$.  The social multiplier effects appear as a product, suggesting it is the reciprocation of the peer effects that matter.  If $d_1$ is small, $d_2$ has to be that much larger to render the steady state unstable.\\

We can push this further without too much complexity by considering the case with $N$ people, who influence their peers with $d_1$ and are influenced by their peers with $d_2$.  It can be shown that in that case, the eigenvalue that guarantees stability when sufficiently small depends on a product of $d_1$ and $d_2$.  Finally, the most general case, in which every peer effect is allowed to be different, lies at the heart of network theory.  I will not rehash complex results that are outside the scope of this paper, but generally speaking, this intuition holds: the stability of the (0,0) steady state depends on the strength of social multiplier effects, and particularly on ``closed loops" of strong peer effects.  That is, person 1 influences person 2 who influences person 1, or person 1 influences person 2 who influences person 3 who influences person 1, and so on.


  \end{proof}

\Reinforcementshort*

\begin{proof}

Consider a function

$$S: (X(t))_{t \geq 0} \rightarrow (Y(t))_{t \geq 0}$$

acting on an asymptotically equitable random process $(X(t))_{t \geq 0}$ in such a way to produce reinforcement processes.  Then by definition, for some $s<t$, there exist $c_j, b_k > 0$, s.t.  


\begin{align}
y_{ij}(t) &= \delta^{t-s} y_{ij}(s)+ b_k y_{ik}(s)+\epsilon_{ij}(t)\\
y_{ik}(t) &= \delta^{t-s} y_{ik}(s)+ c_j y_{ij}(s)+\epsilon_{ik}(t)
\end{align}





Suppose without loss of generality that an injustice occurs at time $0$, inducing inequities $y_{ik}(0)=x_{ik}(0)>0$ and $y_{ij}(0)=x_{ij}(0)>0$.  Then

\begin{align}
E[y_{ik}(t)] =  \delta^{t} E[y_{ik}(0)]+c_j E[y_{ij}(0)]+E[\epsilon_{ik}(t)]\\
=  \delta^{t} E[x_{ik}(0)]+c_j E[x_{ij}(0)]\\
>\delta^{t} E[x_{ik}(0)]\\
=E[x_{ik}(t)]
  \end{align}
  
 Therefore $S$ is discriminatory over the interval $(0,t)$.  Furthermore, consider $T>t$.
 
 \begin{align}
E[y_{ij}(T)] =  \delta^{T-t} E[y_{ij}(t)]+b_k E[y_{ik}(t)]+E[\epsilon_{ij}(T)]\\
=  \delta^{T-t} \left( \delta^t E[y_{ij}(0)]+ b_k E[y_{ik}(0)] \right) +b_k \left( \delta^{t} E[y_{ik}(0)]+c_j E[y_{ij}(0)] \right)\\
=  \delta^T E[y_{ij}(0)]+ \delta^{T-t} b_k E[y_{ik}(0)] +b_k \left( \delta^{t} E[y_{ik}(0)]+c_j E[y_{ij}(0)] \right)\\
> \delta^T E[y_{ij}(0)]\\
=\delta^T E[x_{ij}(0)]\\
=E[x_{ik}(T)]
  \end{align}
 
 Therefore $S$ is discriminatory over the interval  $(t,T)$.  Altogether, $S$ is discriminatory over the interval $(0,T)$.

\end{proof}

\subsection{Additional examples of Reinforcement Processes}

There are several kinds of reinforcement processes that have been used to explain discrimination and segregation in the literature.  Suppose, for example, that negative racial attitudes drove White Americans to move away when Black Americans moved north during the early and mid-20th century (Rothstein, R., 2017).  Because this resulted in a sharp decrease in resources and the quality of infrastructure at the neighborhood level, this worsened the material conditions of Black Americans and their children (Derenoncourt, 2022), confirming negative Black stereotypes behind negative White racial attitudes.  Note that in this case, negative racial beliefs can actually impact the real world in a way that (avoidably!) confirms the beliefs.  Inequities are sustained over time because they confirm beliefs that reinforce them in the real world.  Importantly, note that the equilibrium selection from initial conditions was determined by injustices-hence these dynamics are not systemic and benign; they are systemic and unjust.

Some reinforcement processes involve beliefs that don’t change the real world, but rather, what is observed.  Consider Loury’s Taxi Cab problem (Loury, 2009).  A taxi driver has an initially negative racial belief-that Black people tend to rob taxi drivers.  As a result, she generally refuses to give rides to Black passengers.  In equilibrium, the majority of Black passengers, with no intention of robbing taxi drivers, will rationally no longer attempt to get taxis.  Instead, only those with the intention of robbing taxi drivers will do so.  The taxi driver will now only observe Black thieves.  Her beliefs are unjust and incorrect in equilibrium, but she cannot correct them because they have changed the data she gets to observe.  A more extreme version of beliefs changing what is observed occurs when beliefs stop one from making observations altogether.  Suppose initial negative beliefs about Black Americans (for example, that they are usually unpleasant to be around) lead White Americans to move their families to all-White neighborhoods, work with all-White co-workers and have only White friends.  Then no new data on Black Americans gets observed anymore, and so the erroneous beliefs will never even be given a chance to be falsified.  Once again, but in distinct ways, initial inequities are allowed to persist over time.

\newpage

\section{References}

\singlespacing

\noindent Ager, P., Boustan, L., and Eriksson, K. (2021). The intergenerational effects of a large wealth shock: white southerners after the Civil War. American Economic Review, 111(11), 3767-3794.\\

\noindent Allen, T., and Donaldson, D. (2020). Persistence and path dependence in the spatial economy (No. w28059). National Bureau of Economic Research.\\

\noindent Allison, P.D., Williams, R., and Moral-Benito, E. (2017). Maximum likelihood for cross-lagged panel models with fixed effects. Socius\\

\noindent Akhmerov, A., and Marbán, E. (2020). COVID-19 and the heart. Circulation research, 126(10), 1443-1455.\\

\noindent Arthur, W. B. (2018). The economy as an evolving complex system II. CRC Press\\

\noindent Ba, B. A., Knox, D., Mummolo, J., and Rivera, R. (2021). The role of officer race and gender in police-civilian interactions in Chicago. Science, 371(6530), 696-702\\

\noindent Bagenstos, S. R. (2014). Bottlenecks and Antidiscrimination Theory; Bottlenecks: A New Theory of Equal Opportunity by Joseph Fishkin\\

\noindent Bailey, Z. D., Krieger, N., Agénor, M., Graves, J., Linos, N., and Bassett, M. T. (2017). Structural racism and health inequities in the USA: evidence and interventions. The lancet, 389(10077), 1453-1463\\

\noindent Bolte, L., Immorlica, N., and Jackson, M. O. (2020). The role of referrals in immobility, inequality, and inefficiency in labor markets. arXiv preprint arXiv:2012.15753.\\

\noindent Brock, W. A., and Durlauf, S. N. (2001). Discrete choice with social interactions. The Review of Economic Studies, 68(2), 235-260.\\

\noindent Billings, S. B., Deming, D. J., and Ross, S. L. (2019). Partners in crime. American Economic Journal: Applied Economics, 11(1), 126-150.\\

\noindent Blume, L. E., and Durlauf, S. N. (2005). Identifying social interactions: A review. Madison: Social Systems Research Institute, University of Wisconsin.\\

\noindent Bohren, J. A., Hull, P., and Imas, A. (2022). Systemic discrimination: Theory and measurement (No. w29820). National Bureau of Economic Research.\\

\noindent Calvin, R., Winters, K., Wyatt, S. B., Williams, D. R., Henderson, F. C., and Walker, E. R. (2003). Racism and cardiovascular disease in African Americans. The American journal of the medical sciences, 325(6), 315-331\\

\noindent Charles, K. K., and Guryan, J. (2011). Studying discrimination: Fundamental challenges and recent progress. Annu. Rev. Econ., 3(1), 479-511\\

\noindent Cook, L. D. (2020). Policies to broaden participation in the innovation process. Policy Proposal, The Hamilton Project, Brookings Institution, Washington, DC.\\

\noindent Cunningham, J. P., and Lopez, J. J. (2021, May). Civil Rights Enforcement and the Racial Wage Gap. In AEA Papers and Proceedings (Vol. 111, pp. 196-200). 2014 Broadway, Suite 305, Nashville, TN 37203: American Economic Association. \\

\noindent  Darity Jr, W. A. (2022). Position and possessions: Stratification economics and intergroup inequality. Journal of Economic Literature, 60(2), 400-426.\\

\noindent Darity Jr, W. A., and Mullen, A. K. (2022). From here to equality: Reparations for Black Americans in the twenty-first century. UNC Press Books.\\

\noindent Derenoncourt, E., Kim, C. H., Kuhn, M., and Schularick, M. (2022). Wealth of two nations: The US racial wealth gap, 1860-2020 (No. w30101). National Bureau of Economic Research\\

\noindent Derenoncourt, E. (2022). Can You Move to Opportunity? Evidence from the Great Migration. The American Economic Review, 112(2), 369-408.\\

\noindent Durlauf, S. N. (2012). Complexity, economics, and public policy. Politics, Philosophy, and Economics, 11(1), 45-75.\\

\noindent Durlauf, S. N., and Ioannides, Y. M. (2010). Social interactions. Annu. Rev. Econ., 2(1), 451-478.\\

\noindent Feagin, J. (2013). Systemic racism: A theory of oppression. Routledge\\

\noindent Fincher, C., Williams, J. E., MacLean, V., Allison, J. J., Kiefe, C. I., and Canto, J. (2004). Racial disparities in coronary heart disease. Ethnicity and disease, 14(3), 360-371\\

\noindent Fryer Jr, R. G., and Loury, G. C. (2005). Affirmative action and its mythology. Journal of Economic Perspectives, 19(3), 147-162.\\

\noindent Glaeser, E. L., Sacerdote, B. I., and Scheinkman, J. A. (2003). The Social Multiplier. Journal of the European Economic Association, 1(2-3), 345-353.\\

\noindent Guryan, J., and Charles, K. K. (2013). Taste‐based or statistical discrimination: the economics of discrimination returns to its roots. The Economic Journal, 123(572), F417-F432.\\

\noindent Hamilton, D., and Darity Jr, W. (2010). Can ‘baby bonds’ eliminate the racial wealth gap in putative post-racial America?. The Review of Black Political Economy, 37(3-4), 207-216.\\

\noindent Jackson, J. W. (2021). Meaningful causal decompositions in health equity research: definition, identification, and estimation through a weighting framework. Epidemiology (Cambridge, Mass.), 32(2), 282.\\

\noindent Johnson, R. C., and Jackson, C. K. (2019). Reducing inequality through dynamic complementarity: Evidence from Head Start and public school spending. American Economic Journal: Economic Policy, 11(4), 310-349\\

\noindent Kurtulus, F. A. (2016). The impact of affirmative action on the employment of minorities and women: a longitudinal analysis using three decades of EEO‐1 filings. Journal of Policy Analysis and Management, 35(1), 34-66.\\

\noindent Lang, K., and Spitzer, A. K. L. (2020). Race discrimination: An economic perspective. Journal of Economic Perspectives, 34(2), 68-89.\\

\noindent Loury, G. C. (2009). The anatomy of racial inequality. Harvard University Press.\\

\noindent McMillon, D. (2024). The Self-Reinforcing Effects of Temporary Interventions: Systems Thinking to Reduce Systemic Disadvantage. Available at SSRN\\

\noindent McMillon, D., Simon, C. P., and Morenoff, J. (2014). Modeling the underlying dynamics of the spread of crime. PloS one, 9(4), e88923.\\

\noindent Meschede, T., Darity Jr, W., and Hamilton, D. (2015). Financial resources in kinship and social networks: Flow and relationship to household wealth by race and ethnicity among Boston residents. Federal Reserve Bank of Boston Community Development Discussion Paper, (2015-02).\\

\noindent Miller, C. (2017). The persistent effect of temporary affirmative action. American Economic Journal: Applied Economics, 9(3), 152-190.\\

\noindent Miller, M. C. (2020). “The Righteous and Reasonable Ambition to Become a Landholder”: Land and Racial Inequality in the Postbellum South. Review of Economics and Statistics, 102(2), 381-394\\

\noindent Kline, P., Rose, E. K., and Walters, C. R. (2022). Systemic discrimination among large US employers. The Quarterly Journal of Economics, 137(4), 1963-2036\\

\noindent O'Brien, R. L. (2012). Depleting capital? Race, wealth and informal financial assistance. Social Forces, 91(2), 375-396

\noindent Patton, D. U., Hong, J. S., Williams, A. B., and Allen-Meares, P. (2013). A review of research on school bullying among African American youth: an ecological systems analysis. Educational Psychology Review, 25, 245-260.\\

\noindent Pfeffer, F. T., and Killewald, A. (2015). How rigid is the wealth structure? intergenerational correlations of family wealth. Population Studies Center, University of Michigan.\\

\noindent Phelan, J. C., Link, B. G., and Tehranifar, P. (2010). Social conditions as fundamental causes of health inequalities: theory, evidence, and policy implications. Journal of health and social behavior, 51, S28-S40\\

\noindent Phelan, J. C., and Link, B. G. (2015). Is racism a fundamental cause of inequalities in health?. Annual Review of Sociology, 41, 311-330.

\noindent Powell, J. A. (2007). Structural racism: building upon the insights of John Calmore. NCL Rev., 86, 791.\\

\noindent Reskin, B. (2012). The race discrimination system. Annual review of sociology, 38, 17-35.\\

\noindent Rothstein, R. (2017). The color of law: A forgotten history of how our government segregated America. Liveright Publishing\\

\noindent Roithmayr, D. (2014). Reproducing racism. In \textit{Reproducing Racism}. New York University Press.\\

\noindent Roithmayr, D. (2016). The dynamics of excessive force. U. Chi. Legal F., 407.\\

\noindent Schell, C. J., Dyson, K., Fuentes, T. L., Des Roches, S., Harris, N. C., Miller, D. S., ... and Lambert, M. R. (2020). The ecological and evolutionary consequences of systemic racism in urban environments. Science, 369(6510), eaay4497.\\

\noindent Sen, M., and Wasow, O. (2016). Race as a bundle of sticks: Designs that estimate effects of seemingly immutable characteristics. Annual Review of Political Science, 19, 499-522.\\

\noindent Spencer, S. J., Logel, C., and Davies, P. G. (2016). Stereotype threat. Annual review of psychology, 67, 415-437.\\

\noindent Wadhera, R. K., Figueroa, J. F., Rodriguez, F., Liu, M., Tian, W., Kazi, D. S., ... and Joynt Maddox, K. E. (2021).  Racial and ethnic disparities in heart and cerebrovascular disease deaths during the COVID-19 pandemic in the United States. Circulation, 143(24), 2346-2354.\\

\noindent Whitehead, M. and Dahlgren, G. (2006). Concepts and principles for tackling social inequities in health: Leveling up Part 1. World Health Organization: Studies on social and economic determinants of population health, 2, 460-474.\\

\noindent Wodtke, G.T. (2018). Regression-based adjustment for time-varying confounders.
Sociological Methods and Research\\

\noindent Zivin, J. S. G., and Singer, G. (2023). Disparities in pollution capitalization rates: The role of direct and systemic discrimination (No. w30814). National Bureau of Economic Research.

\end{document}